\documentclass[11pt]{article}
\usepackage{lmodern}
\usepackage{tcolorbox}
\usepackage{soul}
\usepackage[normalem]{ulem}

\usepackage[margin=1in]{geometry}
\usepackage{amsmath,array,bm}
\usepackage{amsmath,amsfonts}
\usepackage{amsthm}
\usepackage{amssymb,latexsym}
\usepackage{algorithm}
\usepackage{subcaption}
\usepackage{algorithmicx}
\usepackage{marginnote}
\usepackage[]{algpseudocode}
\usepackage{float}
\usepackage{thmtools}
\usepackage{mathtools}
\usepackage{thm-restate}
\usepackage{bussproofs}
\usepackage[square,sort,comma,numbers]{natbib}
\usepackage{tikz}
\usepackage{xspace}
\usepackage{xcolor}
\usepackage{tipa}
\usepackage{mdframed}
\usepackage{multirow}
\usepackage{hhline}
\usepackage{todonotes}
\usepackage{tcolorbox} 
\usepackage{yhmath}
\usepackage{enumitem}
\usepackage{graphicx}
\usepackage{hyperref}

\newtheorem{lemma}{Lemma}
\newtheorem{claim}{Claim}
\newtheorem{observation}{Observation}
\newtheorem{theorem}{Theorem}

\newtheorem{definition}{Definition}

\algtext*{EndWhile}
\algtext*{EndIf}
\algtext*{EndFor}

\providecommand{\email}[1]{\href{mailto:#1}{\nolinkurl{#1}\xspace}}
\definecolor{mahogany}{rgb}{0.75, 0.25, 0.0}
\hypersetup{
     colorlinks  = true,
     citecolor   = mahogany,
	 linkcolor	= olive
}

\title{A Polynomial-Time Approximation for Pairwise Fair $k$-Median Clustering} 
 \author{Sayan Bandyapadhyay\thanks{Portland State University. Email: \email{sayanb@pdx.edu }. Supported by a grant (2311397) from the National Science Foundation.} \and Eden Chlamt\'a\v{c}\thanks{Ben Gurion University. Email: \email{chlamtac@cs.bgu.ac.il}. The work was done while the author was visiting and supported by TTIC.}
 \and Zachary Friggstad\thanks{University of Alberta. Email: \email{zacharyf@ualberta.ca}. Supported by an NSERC Discovery Grant.} 
 \and Mahya Jamshidian\thanks{University of Alberta. Email: \email{mjamshidian@ualberta.ca}}
 \and Yury Makarychev\thanks{Toyota Technological Institute at Chicago (TTIC). Email: \email{yury@ttic.edu}. Supported by NSF Awards CCF-1955173, CCF-1934843, and ECCS-2216899.}
 \and Ali Vakilian\thanks{Toyota Technological Institute at Chicago (TTIC). Email: \email{vakilian@ttic.edu}}}
\date{} 

\theoremstyle{definition}

\newcommand{\frc}{\textsc{Fair Representational Clustering}\xspace}
\newcommand{\pfc}{\textsc{Pairwise Fair Clustering}\xspace}
\newcommand{\pfkmd}{\textsc{Pairwise Fair $k$-Median}\xspace}

\begin{document}

\maketitle

\begin{abstract}
In this work, we study pairwise fair clustering with $\ell \ge 2$ groups, where for every cluster $C$ and every group $i \in [\ell]$, the number of points in $C$ from group $i$  must be at most $t$ times the number of points in $C$ from any other group $j \in [\ell]$, for a given integer $t$. 
To the best of our knowledge, only bi-criteria approximation and exponential-time algorithms follow for this problem from the prior work on fair clustering problems when $\ell > 2$. In our work, focusing on the $\ell > 2$ case, we design the first polynomial-time $O(k^2\cdot \ell \cdot t)$-approximation for this problem with $k$-median cost that does not violate the fairness constraints. We complement our algorithmic result by providing hardness of approximation results, which show that our problem even when $\ell=2$ is almost as hard as the popular uniform capacitated $k$-median, for which no polynomial-time algorithm with an approximation factor of $o(\log k)$ is known. 
\end{abstract}

\section{Introduction}\label{sec:intro} 
Clustering is a fundamental task in theoretical computer science and machine learning aimed at dividing a set of data items into several groups or clusters, such that each group contains similar data items. Typically, the similarity between data items is measured using a metric distance function. Clustering is often modeled as an optimization problem where the objective is to minimize a global cost function that reflects the quality of the clusters; this function varies depending on the application. Among the many cost functions studied for clustering, the most popular are $k$-median, $k$-means, and $k$-center. These objectives generally aim to minimize the variance within the clusters, serving as a proxy for grouping similar data items 

In this work, we study clustering problems with fairness constraints, commonly known as fair clustering problems. Fair clustering emerged as one of the most active research areas in algorithms motivated by the recent trend of research on fairness in artificial intelligence. In a seminal work, Chierichetti et al.~\cite{chierichetti2017fair} introduced a fair clustering problem, where given a set $R$ of red points, a set $B$ of blue points, and an integer balance parameter $t \ge 1$, a clustering is said to be \textit{balanced} if, in every cluster, the number of red points is at least $1/t$ times the number of blue points and at most $t$ times the number of blue points. Then, the goal is to compute a balanced clustering so that a certain cost function is minimized (e.g., $k$-median). In their work, they obtained a 4-approximation for balanced $k$-center and an $O(t)$-approximation for balanced $k$-median. Both algorithms run in polynomial time. Subsequently, Schmidt~{et al.}~\cite{schmidt2019fair} gave an $n^{O(k/\epsilon)}$ time $(1+\epsilon)$-approximation for the Euclidean balanced $k$-means and $k$-median. In the same Euclidean case, Backurs et al.~\cite{backurs2019scalable} designed a near-linear
time $O(d\cdot\log n)$-approximation, where $d$ is the dimension of the input space. 

Since the work of \cite{chierichetti2017fair}, many researchers have generalized the notion of balance to multiple colors (or groups), where the input point set $P$ is partitioned into $\ell\ge 2$ groups $P_1,\ldots, P_{\ell}$. R{\"o}sner et al.~\cite{rosner2018privacy} studied the problem where a clustering is called fair if, for every cluster $C$ and every group $1\le i\le \ell$, the ratio $\frac{|P_i\cap C|}{|C|}$ is equal to the ratio $\frac{|P_i|}{|P|}$. They obtained a polynomial-time $O(1)$-approximation for this problem with $k$-center cost. B{\"o}hm et al.~\cite{bohm2020fair} considered a restricted case of this problem where $|P_i|=|P|/\ell$ for all $i$, i.e., every cluster must contain the same number of points from each group. They gave a polynomial-time $O(1)$-approximation for this problem with $k$-median cost. Later, Bercea et al.~\cite{bercea2019cost} and Bera et al.~\cite{bera2019fair} independently defined the following generalization of balanced clustering with multiple groups.
\begin{definition}[\frc]
We are given $\ell$ disjoint groups $P_1,\ldots, P_{\ell}$ containing a total of $n$ points within a metric space $(\Omega,d)$. Each group $P_i$, for $1 \leq i \leq \ell$, is associated with balance parameters $\alpha_i, \beta_i \in [0,1]$. A clustering is called \emph{fair representational} if, in each cluster, the fraction of points from group $i$ is at least $\alpha_i$ and at most $\beta_i$, for all $1 \leq i \leq \ell$. Then, the goal is to compute a fair representational clustering that minimizes a given cost function. 
\end{definition}

Bercea et al.~\cite{bercea2019cost} designed polynomial-time $O(1)$-approximations for fair representational $k$-median and $k$-means, albeit with a bi-criteria kind of guarantee where the solution clustering may violate the fairness constraints by an additive $\pm 1$ term. Independently, Bera et al.~\cite{bera2019fair} obtained similar results, for a more general $\ell_p$ cost, with an additive $\pm 3$ violation. 
Schmidt et al.~\cite{schmidt2019fair}, studied composable coreset construction for these problems and designed streaming algorithms for fair representational clustering. 
Subsequently, Dai et al.~\cite{dai2022fair} proposed an $O(\log k)$-approximation with $k$-median cost that runs in time $n^{O(\ell)}$. Concurrently to~\cite{bercea2019cost,bera2019fair}, Ahmadian et al.~\cite{ahmadian2019clustering} studied the fair representational $k$-center with only the upper bound requirement (i.e., $\alpha_i=0$, $\forall i\in [\ell]$). They devised a bi-criteria approximation with an additive $\pm 2$ violation. 

Next, we define another natural generalization of balanced clustering \cite{chierichetti2017fair} considered in \cite{abs-2007-10137}, which is the main focus of our work. 
\begin{definition}[\pfc]
We are given $\ell$ disjoint groups $P_1,\ldots, P_{\ell}$ containing a total of $n$ points within a metric space $(\Omega,d)$, and an integer balance parameter $t \ge 1$. A clustering is called \emph{pairwise fair} if, for any cluster $C$ and any pair of groups $1\le i,j\le \ell$, the number of points from group $i$ in $C$ is at least $1/t$ times the number of points from group $j$ in $C$ and at most $t$ times the number of points from group $j$ in $C$. Then, the goal is to compute a pairwise fair clustering that minimizes a given cost function.
\end{definition}

Note that one of the main differences between \frc and \pfc is that in the former, the balance parameters can be different for different groups, whereas, in the latter, we have a uniform parameter $t$ for all groups (see Appendix for a demonstration). Intuitively, this might give the impression that \pfc is computationally easier than \frc. Indeed, this intuition is correct when the number of groups $\ell$ is exactly 2. {In this case, \pfc is a restricted version of \frc.} To see this, set $\alpha_1 = \alpha_2 = \frac{1}{t+1}$ and $\beta_1 = \beta_2 = \frac{t}{t+1}$.

\begin{figure}[!t]
\centering
\includegraphics[width=13cm]{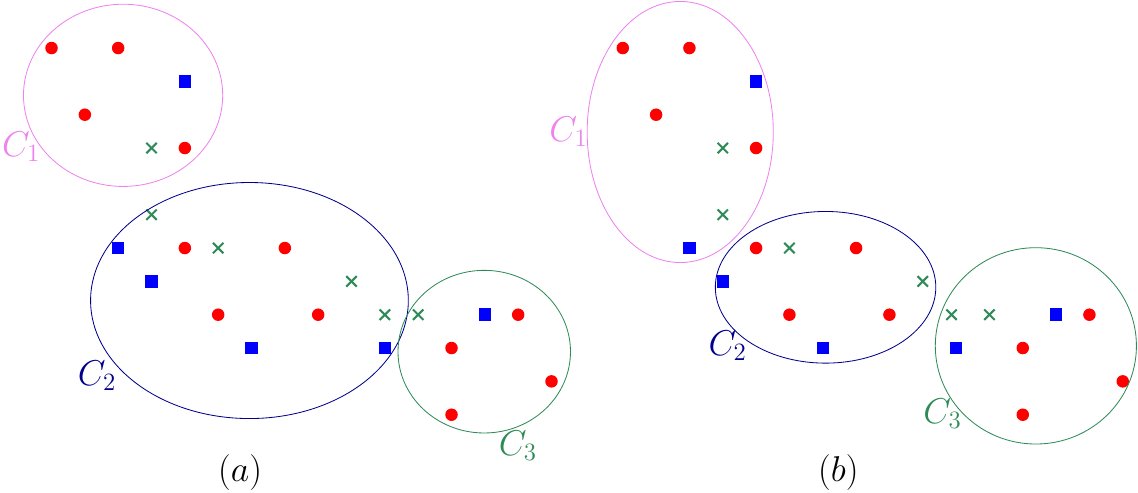}
\caption{\small (a) A fair representational clustering for three groups red (r, disk), blue (b, square), and green (g, cross). Here $\alpha_r=1/3,\beta_r=2/3,\alpha_b=1/6,\beta_b=1/3,\alpha_g=1/6,\beta_g=1/3$. $C_1,C_2,C_3$ are the clusters. (b) A pairwise fair clustering with $t=2$ for the same dataset. The clustering in (a) is not pairwise fair for $t=2$, as $C_1$ contains 4 red (disk) points, but only 1 blue (square).}
\label{fig:1.5}
\end{figure}

{Another special case where the two problems become equivalent is when the clusters are required to be fully balanced}, i.e., \frc with $\alpha_i=\beta_i=\frac{1}{\ell}$ $\forall i$ and \pfc with $t=1$. 
However, for $\ell \ge 3$ and arbitrary parameters, it is not clear how to establish such direct connections between the two problems.

Bandyapadhyay et al.~\cite{abs-2007-10137} considered the fixed-parameter tractability (FPT) of a collection of fair clustering problems. They obtained $O(1)$-approximations for \frc and \pfc\footnote{They study \pfc by the name of $(t,k)$-fair clustering.} (with $k$-median and $k$-means objectives) that run in time $(k\ell)^{O(k\ell)}\cdot n^{O(1)}$ and $(k\ell)^{O(k)}\cdot n^{O(1)}$, respectively.  
We note that, using the apparent connection between the notions of fairness in \frc and \pfc, the bi-criteria approximations of \cite{bera2019fair,bercea2019cost} can be suitably adapted for the latter problem.   

It is evident from the above discussion that despite the extensive research on fair $k$-median with multiple groups, the known approximation algorithms are either bi-criteria that violate fairness constraints or have super-polynomial running time. Even when $t = O(1)$, no polynomial-time approximation algorithm is known for \pfc.  

\paragraph{Our Contributions.} In this work, we consider \pfc with $k$-median cost (\pfkmd). For a pairwise fair clustering $\{X_1,\ldots,X_k\}$ with cluster centers $\{c_1,\ldots,c_k\}\subset \Omega$, the $k$-median cost is $\sum_{j=1}^k \sum_{p\in X_j} d(p,c_j)$. The goal is to find a pairwise fair clustering along with a set of centers that minimizes the cost. We design a true polynomial-time approximation for this problem, i.e. the solutions do not violate the fairness constraints at all.  

\begin{theorem}\label{thm:algopfc}
    There is a polynomial-time $O(k^2\cdot \ell \cdot t)$-approximation algorithm for \pfkmd.
\end{theorem}

As with many previous algorithms for fair clustering models, our algorithm begins by using a constant-factor approximation algorithm for standard \textsc{$k$-Median} to select the set of $k$ centers $F$. We show there is an $O(1)$-approximation for the \pfkmd instance that uses centers $F$, so the focus is then to find a fair assignment to these centers. We use a natural assignment linear programming (LP) relaxation for this task, the main purpose is to find initial (perhaps fractional) values indicating how many clients from each group $P_i$ should be assigned to each center. The LP also serves to help us decompose the problem into instances where the distance between any two points is $O(k)$ times the optimum solution cost.

We then find an integer assignment of clients that almost matches each of these initial values, with only slight additive violation in the fairness constraint at each center. So far, our assignment is still within a constant-factor of the optimum solution cost. Our final step is to show that we can move some points from their currently-assigned to different centers. The number of points moved is a polynomial function of $\ell, k$ and $t$, hence our approximation guarantee.

While the final theoretical approximation guarantee is independent of the number of points, it still may seem a bit large. One might expect in practice that the number of points to be clustered is considerably larger than the parameters $\ell, k$, and $t$. In such cases, we would expect the cost increase for reassigning this bounded number of points in the final step to be small compared to the initial \textsc{$k$-median} solution. So we expect our algorithm to perform much better than this worst-case guarantee in practice.

To investigate this notion, we implemented and tested our algorithm on a variety of datasets to not only test if the algorithm performed better in practice, but also to check the ``cost of fairness'' using our approach: namely how much more expensive is our solution compared with unfair clustering solution at the start of our algorithm. Some of these results appear in Section~\ref{sec:experiments}.

\medskip
We also provide two hardness of approximation results in Section \ref{sec:hardness}. First, we show that \pfkmd is as hard as the popular capacitated $k$-median, up to a factor of $(1+\epsilon)$ for any $\epsilon > 0$. Despite extensive research on this problem, the best-known polynomial-time approximation factor remains to be $O(\log k)$ \cite{DBLP:conf/esa/AdamczykBMM019}. Our polynomial-time reduction demonstrates that if \pfkmd admits a $o(\log k)$-approximation, so does capacitated $k$-median.

Second, we study the complexity of the variant of \pfkmd when the groups $P_1,\ldots,P_{\ell}$ are not disjoint, i.e., a point can be included in multiple groups. Note that the bi-criteria approximation of Bera et al.~\cite{bera2019fair} for \frc also works for this variant, albeit with a larger additive violation. We prove that \pfkmd with overlapping groups is NP-hard to approximate within a factor of $n^{1-\varepsilon}$ for every $\varepsilon > 0$, even when $k=2$. 
We remark this is the first hardness result (with disjoint groups) holds also for $\frc$ with $k$-median cost.

\subparagraph*{Other Related Work.} The Euclidean version of fair $k$-means and $k$-median have also received attention, where points are from $\mathbb{R}^{d}$. Schmidt~{et al.}~\cite{schmidt2019fair} gave an $n^{O(k/\epsilon)}$ time $(1+\epsilon)$-approximation for balanced clustering with two groups. B\"ohm~{et al.}~\cite{bohm2020fair} obtained an $n^{\text{f}(k/\epsilon)}$ time $(1+\epsilon)$-approximation for \pfc in a restricted setting, where $f$ is a polynomial function. Lastly, Backurs et al.~\cite{backurs2019scalable} designed a near-linear time $O(d\cdot  \log n)$-approximation when points are in $\mathbb{R}^d$. 

The clustering problem has been explored with various other fairness notions, including fair center representation~\cite{krishnaswamy2011matroid,chen2016matroid,krishnaswamy2018constant,kleindessner2019fair,chiplunkar2020solve,jones2020fair,hotegni2023approximation}, proportional fairness~\cite{chen2019proportionally,micha2020proportionally}, min-max fairness~\cite{abbasi2020fair,GhadiriSV21,makarychev2021approximation,chlamtavc2022approximating,ghadiri2022constant,gupta2022lp}, fair colorful~\cite{bandyapadhyay2019constant,jia2020fair,anegg2020technique}, and individual fairness notions~\cite{jung2019center,mahabadi2020individual,negahbani2021better,vakilian2022improved,brubach2020pairwise,brubach2021fairness,ahmadi2022individual,aamand2023constant,mosenzon2024scalable}.

\section{The Approximation Algorithm}
\label{sec:approx}



In this section, we prove Theorem \ref{thm:algopfc} by designing an algorithm for \pfkmd. The steps will be described in the following subsections and the entire algorithm will be summarized at the end of this section. We assume $t \geq 2$ as the case $t = 1$ has already been studied and has a constant-factor approximation~\cite{bohm2020fair}.

We begin by introducing some notation. 

Let $[\ell] = \{1, 2, \ldots, \ell\}$ and $P = \cup_{a \in [\ell]} P_a$. For any assignment $\sigma: P \rightarrow F$ for a subset $F \subset \Omega$, its cost, denoted by $cost(\sigma)$, is $\sum_{v \in P} d(v,\sigma(v))$. Moreover, $\sigma$ is called \textit{fair} if for any $c \in F$ and $a,b \in [\ell]$, $|\sigma^{-1}(c)\cap P_a| \le t\cdot |\sigma^{-1}(c)\cap P_b|$. Let $OPT$ be the optimal cost for the \pfkmd instance, $F^*=\{c_1^*,\ldots,c_k^*\}$ be a corresponding optimal set of centers, and $\{C_1^*,\ldots,C_k^*\}$ be a corresponding optimal fair clustering. 

\begin{observation}\label{obs:t-balance}
    Consider two $t$-balanced clusters with $b_i$ and $b_i'$ points, respectively from each group $1\le i\le \ell$. Then, the merger of these two clusters is also $t$-balanced. 
\end{observation}
\begin{proof}
Consider any two indexes $i$ and $j$. The merged cluster has $b_i+b_i'$ and $b_j+b_j'$ points, respectively from groups $i$ and $j$. Now, $b_i+b_i'\le t\cdot b_j+t\cdot b_j'= t\cdot (b_j+b_j')$. Hence, the merged cluster is also $t$-balanced.
\end{proof}

Our algorithm has three major steps which we describe below. 

\subsection{Step 1: Computation of vanilla $k$ centers}
To begin, we use any poly-time $\alpha$-approximation algorithm for $k$-median on the point set $P$ to compute a set of $k$ centers $F \subset \Omega$. Let $\sigma: P \rightarrow F$ be the assignment function that maps each point in $P$ to a nearest center in $F$. Currently, the best polynomial-time approximation algorithm for $k$-median is a $2.675$-approximation by Byrka et al. \cite{byrka2014improved}.

The following lemma shows that an approximation algorithm for this \pfc instance can be obtained using $F$ as the centers.
\begin{lemma}\label{lem:approxfair}
    There exists a fair assignment $\gamma: P\rightarrow F$ such that $cost(\gamma)\le (2+\alpha) \cdot OPT$. 
\end{lemma}

\begin{proof}
    Let $\phi: F^* \rightarrow F$ map each $i \in F^*$ to its nearest $i' \in F$. Consider the following assignment $\gamma: P\rightarrow F$. For each $1\le i\le k$, assign all the points in the $i$-th optimal fair cluster $C_i^*$ to $\phi(c_i^*)$.

    First, we prove that $\gamma$ is fair, i.e., for any $c \in F$ and $i,j \in [\ell]$, $|\gamma^{-1}(c)\cap P_i| \le t\cdot |\gamma^{-1}(c)\cap P_j|$. Let $C_{i^1}^*,\ldots,C_{i^{\kappa}}^*$ be the clusters whose points are assigned to $c$ through $\gamma$. Then, the set of points that are assigned to $c$ is the union of these clusters. By applying Observation \ref{obs:t-balance} in an inductive manner, we obtain $|\gamma^{-1}(c)\cap P_i| \le t\cdot |\gamma^{-1}(c)\cap P_j|$. 

    Next, we bound the cost of $\gamma$. Consider any point $p\in P$ and let $p \in C_i^*$. 

    \begin{align*}
        d(p,\gamma(p)) = d(p,\phi(c_i^*))&\le d(p,c_i^*)+d(c_i^*,\phi(c_i^*)) \tag{triangle inequality}\\
        & \le d(p,c_i^*)+d(c_i^*,\sigma(p))\tag{$\phi(c_i^*)$ is the nearest neighbor of $c_i^*$ in $F$}\\
        & \le d(p,c_i^*)+d(c_i^*,p)+d(p,\sigma(p))\tag{triangle inequality}\\
        & = 2d(p,c_i^*)+d(p,\sigma(p))
    \end{align*}
    So, $cost(\gamma)\le \sum_{p\in P} (2d(p,c_i^*)+d(p,\sigma(p)))\le (2+\alpha) \cdot OPT$. 
\end{proof}
The more difficult task in obtaining our \pfc approximation is to find a fair assignment from $P$ to $F$, which is focus of the rest of this section.

\subsection{Step 2: Finding a low-cost, nearly-fair assignment}
We begin by introducing an LP relaxation of the problem of finding a low-cost \pfc solution using centers $F$.
Let $x_{j,i}$ be a variable indicating we move client $j \in P$ to center $i \in F$. 

In the rest of the algorithm, we try all distances $D$ of the form $d(i,j)$ for some $i \in F$ and some $j \in P$. We solve linear program $LP(D)$ below for each such distance $D$, rejecting any for which $LP(D)$ is infeasible.
For each value $D$ where the LP is feasible, we run the rest of the algorithm and return the best \pfc solution found over all values $D$ for which $LP(D)$ is feasible.

    \begin{equation*}[LP(D)]
    \begin{split}
        \text{min} \quad
    \sum_{i\in F} \sum_{j \in P}& d(i,j)\cdot x_{j,i}  \label{LP(D)} \\
    \text{s.t.} \quad
    \sum_{i \in F} x_{j,i}&= 1, \quad \forall j \in P \\
    \sum_{j \in P_{a}}x_{j,i} &\leq t \cdot \sum_{j^{'} \in P_{b}}x_{j^{'},i}, \quad \forall i\in P \,\text{and}\, (a,b) \in [\ell] \times [\ell] \\
    x_{j,i} &= 0, \quad \,\text{ if $d(i,j) > D$} \\
    x_{j,i} &\geq 0, \quad \,\forall j\in P, i\in F
    \end{split}
    \end{equation*}

For any $D$ and any feasible solution $x$ to $LP(D)$, let $G(D,x)$ be the {\em support} graph with vertices $F \cup P$ and edges $\{\{i,j\} : i \in F, j \in P, x_{i,j} > 0\}$.

\begin{lemma}\label{lem:D}
Let $\gamma$ be the assignment from Lemma \ref{lem:approxfair} and let $D = \max_{j \in P} d(j,\gamma(j))$. 
Then $LP(D)$ is feasible and has optimal solution value at most $(2+\alpha) \cdot OPT$. For any $j \in P$ and $i \in F$ in the same connected component $C \subseteq F \cup P$ of $G(D,x)$ (where $x$ is any feasible solution to $LP(D)$) we have $d(i,j) \leq |F \cap C| \cdot D$.
\end{lemma}
\begin{proof}
That $LP(D)$ is feasible and has value at most $(2+\alpha) \cdot OPT$ is simply because the natural $\{0,1\}$-assignment corresponding to $\gamma$ is a feasible solution to the LP.

Now consider any $i \in F$ and $j \in P$ in the same connected component in $G(D,x)$. Let $j = j_1, i_1, j_2, i_2, \ldots, j_b, i_b = i$ be any path in $G(D,x)$ with the minimum number of edge hops. Each hop uses an edge with distance at most $D$ by the LP constraints, so
\[ d(i,j) \leq d(j_1, i_1) + \sum_{a = 2}^{b} d(i_{a-1},j_a) + d(j_b, i_b) \leq |F \cap C| \cdot D \]
where the last bound is because the the path is a minimum-hop path so the centers $i_1, \ldots, i_b$ are distinct.
\end{proof}

From this point on, we will analyze the algorithm for the case $D = \max_{j \in P} d(j, \gamma(j)) \leq (2+\alpha) \cdot OPT$. Since our final algorithm takes the best solution over all $D$, the final solution's cost will be no worse than the cost of the solution in this case. We remark here that our final approximation guarantee will be of the form $O(OPT) + O(k^2 \cdot \gamma \cdot t \cdot D)$. In typical instances, one would expect $D$ to be much smaller than $cost(\gamma)$ so the approximation guarantee is likely to be much better in practice.
 
We now describe a simple procedure that rounds an optimal solution to $LP(D)$ to obtain a low-cost integer assignment of points that is close to a fair assignment. 

\begin{observation}
The restriction of $x^*$ to any connected component $C \subseteq P \cup F$ of $G(x^*, D)$ is a feasible solution to the \pfc instance with clients $P \cap C$. Furthermore, $P \cap C$ is balanced, meaning $|P_a \cap C| \leq t \cdot |P_{a'} \cap C|$ for any two $a, a' \in [\ell]$.
\end{observation}
\begin{proof}
Most statements are straightforward. To see that $P \cap C$ is balanced,
note that the fractional assignment to all centers in $F \cap C$ is balanced.
Since the total number of clients in $P_a \cap C$ for each $a \in [\ell]$ is the sum of their fractional assignments to each center in $F \cap A$. So by Observation \ref{obs:t-balance}, more precisely a straightforward generalization of it to fractional assignments, $P \cap C$ is balanced.
\end{proof}

From now on, we may assume $G(D, x^*)$ is connected. If it was not, we run the following algorithm on the restriction of $x^*$ to each connected component and output the union of the solutions. For each component $C$ let $\nu_C$ be the value of $x^*$ restricted to $C$. Our algorithm will find a solution with cost only $O(|F \cap C|^2 \cdot \ell \cdot t \cdot D)$ more than $\nu_C$ in each component $C$. Adding this for each component $C$ shows the final solution cost would be at most $O(k^2 \cdot \ell \cdot t \cdot D) + (2+\alpha) \cdot OPT = O(k^2 \cdot \ell \cdot t \cdot OPT$), as required.



The initial nearly-fair solution is obtained from an optimal solution $x^*$ as follows:
\begin{itemize}
    \item First, for each $i \in F$, set $\ell_i = \min_{a \in [\ell]} \sum_{j \in P_a} x^*_{j,i}$. Here, $\ell_i$ represents the smallest amount that any group of clients is fractionally assigned to center $i$.
    \item Then find a min-cost integral assignment $\sigma_{int} : P \rightarrow F$, where $|\sigma_{int}^{-1}(i)\cap P_a| \in [\lfloor l_{i} \rfloor, \lceil t\cdot l_{i} \rceil]$.
\end{itemize}



\begin{lemma}
    Such assignment $\sigma_{int}$ exists and has cost at most the optimum solution value of $LP(D)$.
\end{lemma}
\begin{proof}
Consider a bipartite graph with nodes $P$ on the left side and nodes $\{(i,c) : i \in F, a \in [\ell]\}$ on the right (i.e. a node for each group of clients for each center in $F$). For each $j \in P$ with, say, $j \in P_a$ and each $i \in F$ we include an edge $\{j,(i,a)\}$ with cost $d(j,i)$. In this way, $x^*$ describes a fractional assignment such that each $j \in P$ is assigned to an extent of exactly $1$ to nodes on the right side and each node $(i,a)$ on the right has received a total assignment between $\lfloor \ell_i \rfloor$ and $\lceil t \cdot \ell_i \rceil$.

Furthermore, the cost of this fractional assignment is exactly the optimal LP solution cost. By integrality of the bipartite assignment polytope with integer upper and lower bounds, there is an integer assignment $\sigma_{int}$ of no greater cost. Such an assignment can be computed using any polynomial-time minimum-cost assignment algorithm.
\end{proof}

As a result, $cost(\sigma_{int})$ is still $O(OPT)$ and the assignment $\sigma_{int}$ is nearly fair in that for any $i \in F$ and any $a,b \in [\ell]$ we have $|\sigma^{-1}(i) \cap P_a| \leq t \cdot |\sigma^{-1}(i) \cap P_b| + t$.
The fairness condition is only violated in extreme cases where we have exactly $\lfloor l_{i} \rfloor$ clients of some color sent to center $i$ and strictly more than $t \cdot \lfloor l_{i} \rfloor$ clients of some other color assigned to center $i$. That is, we have $\lfloor t \cdot \ell_i \rfloor \leq t \cdot \lfloor \ell_i \rfloor + t$, so if $i$ receives at least $\lfloor \ell_i \rfloor + 1$ clients of each color then the assignment to $i$ would indeed be fair.

In the remainder of the algorithm, we describe how to change the destination of $O(k \cdot \ell \cdot t)$ clients in the assignment $\sigma_{int}$ to get a truly fair solution. Each client that has its destination changed will be moved to another center in its connected component, so by Lemma \ref{lem:D} the total cost will increase by at most $O(k^2 \cdot \ell \cdot t \cdot D)$.



\subsection{Step 3: Fixing the assignment}\label{sec:fixing}
To modify $\sigma_{int}$ to make it a feasible \pfc solution, we will first unassign a limited number clients to get a fair partial assignment and then reassign them to new locations.

Define a set $S$ of {\bf unassigned} clients as follows: for each $i \in F$ and each color $a$ let $\Delta_{i,a} = \max\{0, |\sigma^{-1}_{int}(i) \cap P_a| - t \cdot \lfloor \ell_i \rfloor\}$ denote the number of clients of color $a$ assigned to $i$ in excess of $t$ times the lower bound used in the computation of $\sigma_{int}$. If $\Delta_{i,a} > 0$, add any $\Delta_{i,a}$ clients from $\sigma^{-1}_{int}(i) \cap P_a$ to $S$.

\begin{claim}
    $|S| \leq k \cdot \ell \cdot t$
\end{claim}
\begin{proof}
By the upper bound imposed when $\sigma$ is computed, we have $|\sigma^{-1}_{int}(i) \cap P_a| \leq \lceil t \cdot \ell_i \rceil$. So
\[ \Delta_{i,a} \leq \lceil t \cdot \ell_i \rceil - t \cdot \lfloor \ell_i \rfloor \leq t. \]
Summing this bound for all $i \in F$ and $a \in [\ell]$ yields claim.
\end{proof}

Observe that if we ignore $S$, each center $i$ receives between $\lfloor \ell_i \rfloor$ and $t \cdot \lfloor \ell_i \rfloor$ clients of each color, so this partial assignment is fair. Our next task is to reassign all clients in $S$ such that the fairness remains preserved. In doing so, we may reassign a limited number of additional clients beyond those in $S$.

Let $\sigma$ be the current partial assignment, initially it is the restriction of $\sigma_{int}$ to $P\setminus S$, i.e. those clients that were not unassigned. We also maintain integer bounds $\ell'_i$ for each center, initially set $\ell'_i = \min_{a \in [\ell]} |\sigma^{-1}(i) \cap P_a|$.

~

\noindent
{\bf Invariant}: We will maintain for each center $i$ and each color $a$ that $|\sigma^{-1} (i) \cap P_a| \in [\ell'_i, t \cdot \ell'_i]$, i.e. that the partial assignment is fair. This is true for the initial partial assignment $\sigma$ by how we unassigned clients and by how we set $\ell'_i$.

~

Let $i^*$ denote a particular center in $F$, this may be arbitrarily chosen but should remain fixed throughout the algorithm. Repeat following while $S \neq \emptyset$. Intuitively, the following steps will check if some unassigned client can be assigned to any center while preserving fairness of the current assignment. If not, it moves some clients from their currently-assigned center to $i^*$ so that some unassigned client can also be assigned to $i^*$ while maintaining fairness.

\begin{itemize}
    \item {\bf First Case}: Formally, if there is some $j \in S$ with, say, $j \in P_a$ and some $i \in F$ such that $|\sigma^{-1}(i) \cap P_a| < t \cdot \ell'_i$, set $\sigma(j) := i$ and remove $j$ from $S$ ($j$ is now assigned).
    \item {\bf Second Case}: Otherwise, pick any $j \in S$ and let $a \in [\ell]$ be such that $j \in P_a$. Observe $|\sigma^{-1}(i^*) \cap P_a| = t \cdot \ell'_{i^*}$ since the previous case did not apply. Let $A = \{b \in [\ell] : |\sigma^{-1}(i^*) \cap P_b| = \ell'_{i^*}\}$ (it could be $A = \emptyset$). For each $b \in A$,
    pick any one $i_b \in F$ with $|\sigma^{-1}(i_b) \cap P_b| > \ell'_{i_b}$.
    Then pick any one $j_b \in P_b$ with $\sigma(j_b) = i_b$ and reassign $\sigma(j_b) := i^*$. After doing so for all $j_b \in A$, set $\sigma(j) := i^*$, remove $j$ from $S$, and update $\ell'_{i^*} = \ell'_{i^*} + 1$.
\end{itemize}
In the second step, it could be that $A = \emptyset$: recall the invariant maintains each center has between $\ell'_i$ and $t \cdot \ell'_i$ clients from each group but we do not necessarily insist that the ``lower bound'' $\ell'_i$ is met. If $A = \emptyset$, then the second step merely increases $\ell'_{i^*}$ so $j$ can be assigned to $i^*$ while maintaining the invariants.

Next, we prove that there is some $i_b \in F$ for each $b \in A$ in the second case. Given this, it is straightforward to verify the invariants hold after each iteration.

\begin{lemma}
Whenever the second case is executed, we can always find a center $i_b$ with $|\sigma^{-1}(i_b) \cap P_b| > \ell'_{i_b}$ for each $b \in A$.
\end{lemma}
\begin{proof}
Suppose otherwise, i.e. that for some $b \in A$ we have $|\sigma^{-1}(i) \cap P_b| = \ell'_i$ for every $i\in F$. Since the first case does not apply, then every $j \in P_b$ must be assigned to some center as any unassigned $j \in P_b$ could be assigned to {\em any} center in this case (note we are using $t > 1$ in this argument). 
So $|P_b| = \sum_{i \in F} \ell'_i$. On the other hand, since there exists a $j \in P_a$ that cannot be assigned to any center (again since the first case does not apply) and since $j$ itself is not yet assigned, $|P_a| \geq 1 + t \cdot \sum_{i \in F} \ell'_i$. This contradicts the feasibility assumption of the given instance, i.e., $|P_a| \leq t \cdot |P_b|$.
\end{proof}

We now bound the total number of clients $j$ with $\sigma(j) \neq \sigma_{int}(j)$ after the algorithm concludes. For any such $j$, either $j \in S$ or $j$ was moved in the second case above in order to increase the lower bound $\ell'_{i^*}$. Each time this lower bound was increased, at most $\ell$ clients were reassigned. So a simple bound on the number of clients reassigned this way is $\ell \cdot |S| \leq k \cdot \ell^2 \cdot t$.

But we can refine this analysis a bit by showing the second case only executes at most $k$ times in total. The following claim is the starting point for this argument.

\begin{claim}\label{claim:move}
Consider an iteration where the second case is being executed and let $j \in S$ be as in the description of this case where, say, $j \in P_a$. No client in $P_a$ that is currently assigned to $i^*$ at the start of this iteration was ever assigned to a different center in any previous iteration.
\end{claim}
\begin{proof}
Suppose otherwise, that some $j' \in P_a$ was earlier assigned to some $i' \neq i^*$ and was reassigned to $i^*$ prior to the iteration being considered. In this earlier iteration the second case was being executed and we had $j \in S$ and $|\sigma^{-1}(i^*) \cap P_a| = \ell'_{i^*} < t \cdot \ell'_{i^*}$ (since $j'$ was moved to $i^*$ in this earlier iteration). This contradicts the fact that $j \in S$ in this earlier iteration and that $j$ could be added to $i^*$ (i.e. the first case could have been executed).
\end{proof}

Towards the refined analysis, if $\ell'_{i^*}$ was never increased then $\sigma(j) \neq \sigma_{int}(j)$ only for $j \in S$. Otherwise, say $\ell'_{i^*}$ was increased $\Gamma \geq 1$ times.

\begin{lemma}
$\Gamma \leq k$.
\end{lemma}
We remark this would prove the number of clients $j_b$ that had $\sigma(j_b)$ change over all iterations of the second case would is at most $k \cdot \ell$, as $|A| \leq \ell$ in each iteration.
So the total number of points $j$ with $\sigma(j) \neq \sigma_{int}(j)$ would then be at most $|S| + k \cdot \ell $ and the final solution cost would then be
\[ cost(\sigma_{int}) + O(k \cdot D) \cdot (|S| + k \cdot \ell) = O(OPT) + O(k^2 \cdot \ell \cdot t \cdot D) = O(k^2 \cdot \ell \cdot t \cdot OPT).\]
That is, the proof of Theorem \ref{thm:algopfc} will be complete.
\begin{proof}
Consider the moment that $\ell'_{i^*}$ was increased the last time (i.e. becomes equal to $\lfloor \ell_{i^*} \rfloor + \Gamma$). At this point, the algorithm was considering some $j \in S$ with, say, $j \in P_a$ to assign to $i^*$.
By Claim \ref{claim:move}, every client $j' \in P_a$ currently assigned to $i^*$ at the start of this iteration was either initially assigned to $i^*$ (i.e. had $\sigma_{int}(j') = i^*$ and was not initially unassigned) or was initially in $S$.
By how $S$ was constructed, it initially contained at most $t \cdot k$ clients in $P_a$ so at this point strictly fewer than $t \cdot k$ clients have been added to $i^*$ (recalling $j \in S \cap P_a$ has not yet been added to $i^*$ at this point).



At the start of the algorithm, we had $|\sigma^{-1}(i^*) \cap P_a| \leq t \cdot \lfloor \ell_i \rfloor$ (using the initial assignment $\sigma$ of $P\setminus S$) and the moment before $\ell'_{i^*}$ is increased to $\Gamma$ we have $|\sigma^{-1}(i^*) \cap P_a| = t \cdot (\lfloor \ell_i \rfloor + \Gamma - 1)$.
So at least $t \cdot (\Gamma-1)$ clients in $S \cap P_a$ have been added to $i^*$ since the start of the algorithm. By the previous paragraph, this also at most $t \cdot k - 1$, i.e. $t \cdot (\Gamma-1) < t \cdot k$.
So $\Gamma-1 < k$, meaning $\Gamma \leq k$ as both values are integers.
\end{proof}

\noindent
The entire algorithm is summarized in Algorithm \ref{alg:one}.
 \begin{algorithm}
 \caption{\pfc approximation.}\label{alg:one}
 \begin{algorithmic}[1]
 \State Use any $O(1)$-approximation for $k$-median to get a set of centers $F$.
 \State $\textsc{CAND} \gets \emptyset$ \Comment{The various feasible solutions found by the algorithm below}.
 \For{each $D \in \{d(j,i) : j \in P, i \in F\}$}
 \State if $LP(D)$ is infeasible, continue to the next iteration of the loop
 \State Let $x^*$ be an optimal solution to $LP(D)$
 \State Let $\ell_i = \min_{a \in [\ell]} \sum_{j \in P_a} x^*_{j,i}$ for each $i \in F$
 \For{each connected component $C$ of $G(D, x^*)$}
 \State Compute a minimum-cost assignment  $\sigma_{int} : P \cap C \rightarrow F \cap C$ satisfying $|\sigma^{-1}_{int}(i) \cap P_a| \in [\lfloor \ell_i \rfloor, \lfloor t \cdot \ell_i \rfloor]$ for each $i \in F \cap C$
 \State Execute the fixing procedure described in Section \ref{sec:fixing} to get a feasible \pfc solution $\sigma$ for this component $C$
 \EndFor
 \State Let $\sigma'$ be the assignment obtained by combining all assignments $\sigma$ found for the components of $G(D, x^*)$.
 \State $\textsc{CAND} \gets \textsc{CAND} \cup \{\sigma'\}$
 \EndFor
 \State \Return the cheapest solution in $\textsc{CAND}$
 \end{algorithmic}
 \end{algorithm}

\section{Hardness of Approximation Results}
\label{sec:hardness}
In this section, we provide our hardness of approximation results for \pfkmd. 
\subsection{With Disjoint Groups}
We show that \pfkmd with disjoint groups is almost as hard as \textsc{Soft Uniform Capacitated $k$-Median} (\textsc{CkM}). 
\begin{definition}[\textsc{Soft Uniform Capacitated $k$-Median}]
We are given metric space $(X, d)$ and parameters $k$ and $u$. Then, the goal is to open $k$ facilities $f_1,\dots, f_k\in X$ and assign each client (point) $x\in X$ to one of the facilities $f_{a(x)}$, via an assignment function $a: X \rightarrow \{f_1, \cdots, f_k\}$, so that at most $u$ clients are assigned to each facility. The objective is to minimize $\sum_{x\in X} d(x, f_{a(x)})$. 
\end{definition}
There is a more general version of this problem, where the sets of client and facility locations may be different. However, Li~\cite{Li14} proved that both variants of the problem admit the same approximation. Further, one can consider a hard variant of \textsc{CkM}, in which all opened facilities must be at different locations. Li~\cite{Li14} showed that if one variant admits an $\alpha$ approximation then the other admits an $O(\alpha)$ approximation (see~\cite{Li14} for details).

\begin{theorem}\label{thm:hardness-disjoint}
Assume that there is a polynomial-time $\alpha$ approximation algorithm for \pfkmd with $\ell = 2$ disjoint groups. Then \textsc{Soft CkM} admits an $(1+ \varepsilon)\alpha$ approximation for every constant $\varepsilon > 0$.
\end{theorem}
\begin{proof}
Consider an instance $(X,d)$ with $k$ and $u$ of \textsc{CkM}. Using standard rescaling and truncating steps, we may assume that the smallest positive distance in $X$ is at least 1 and the diameter $\operatorname{diam}(X)$ of $X$ is at most polynomial in $n = |X|$.

We transform this instance to an instance $\cal I$ of \pfkmd. The set of locations in ${\cal I}$ is $X\cup \{R\}$, where $R$ is an extra point at distance  $\operatorname{diam}(X)$ from all other locations $x\in X$. There will be multiple clients/points at every location, as we will describe below. Instance $\cal I$ has two groups, black $G_b$ and red $G_r$. 
Define $W = \lceil k\cdot \operatorname{diam}(X)/\varepsilon\rceil$.  We put $W$ black points at every location $x\in X$. We put $k$ red points at $R$. The parameter $t$ equals $uW$. 

Now we relate the costs of the \textsc{CkM} and \pfkmd instances.
Consider a solution for \textsc{CkM} with facilities $(f_1,\dots, f_k)$ and assignment $a$. Denote its cost by $z$. Assume Wlog that at least one client is assigned to every opened facility.
We show that there is a corresponding solution for the \pfkmd instance of cost $y\leq Wz + k\cdot \operatorname{diam}(X)$. We simply assign every black point at location $x\in X$ to center $f_{a(x)}$. We assign red point number $i$ at $R$ to facility $f_i$. We open at most $k$ centers. We verify that the fairness constraint holds. Consider a location $v$ with at least one facility. Denote the opened facilities at this location by $f_{i_1}, \dots, f_{i_q}$. Then, we assign $q$ red points to $v$ and $\sum_{j=1}^q W |\{x:a(x) = i_j\}| \leq q W u$ black points to location $f$. The assignment cost is $Wz$ for black points and is at most $k\cdot \operatorname{diam}(X) $ for red points.

Now consider a solution for \pfkmd. We may assume that this solution does not assign any black points to $R$. If it does, we can reassign these black points and an appropriate number of red points to locations in $X$ without increasing the cost, since the assignment cost for black points may only go down (here we use that $d(R,x) = \operatorname{diam}(X)$ for all $x\in X$).

We treat red points assigned to locations in $X$ as facilities. That is, if red point number $i$ is assigned to $x$, we open facility $f_i$ at point $x$ in our solution for \textsc{CkM}. We use the assignment of black points to red points, to define an assignment flow from locations in $X$ to facilities $f_i$ such that
(i) $W$ units of flow leaves every $x\in X$ and (ii) at most $t = uW$ units of flow enters every facility $f_i$.
We divide all flow amounts by $W$ and get a rescaled flow $\Phi$ such that (i$'$) $1$ unit of flow leaves every $x$ and (ii$'$) at most $u$ units of flow enter every $f_i$. The cost of this flow with respect to edge costs $d(u,v)$ is at most $y/W$. Note that if $\Phi$ is integral, then it defines an assignment of points in $X$ to facilities in such a way that at most $u$ points are assigned to each facility. Further, the cost of the assignment is $z \leq y/W$, as required. So if $\Phi$ is integral, we are done. We now show that there is always an \textit{integral} flow $\Phi'$ satisfying properties (i$'$) and (ii$'$) and of cost at most that of $\Phi$. To this end, we construct an $s$-$t$ flow network. The network consists of $4$ layers. The first layer consists of a single vertex $s$, the second layer consists of $X$, the third one consists of $f_1,\dots, f_k$, and the fourth one of a single vertex $t$.
There is an edge of capacity 1 from source $s$ to each vertex $x$ in $X$; there is an edge of infinite capacity from each $x\in X$ to each $f_i$; there is an edge of capacity $u$ from each $f_i$ to $t$. Each edge from $x$ to $f_i$ has cost $d(x,f_i)$; all other edges are free.

Note that $\Phi$ defines a flow from $X$ to $\{f_1,\dots, f_k\}$, which can be uniquely extended to a feasible flow from $s$ to $t$. Its cost is $z\leq y/W$ and it saturates all edges from $s$. Since all the capacities are integral, there exists a feasible integral flow $\Phi'$ of cost at most $z \leq y/W$ that also saturates all edges from $s$. $\Phi'$ defines an integral assignment of points in $X$ to facilities $f_1,\dots, f_k$, as desired.

Now to solve an instance of \textsc{CkM}, we first reduce it to an instance of \pfkmd, run the provided $\alpha$-approximation algorithm, and then transform the obtained solution to a solution of \textsc{CkM}. Instance $\cal I$ of \pfkmd has a solution of cost at most 
$W \cdot OPT + k \cdot \operatorname{diam}(X)$, where $OPT$ is the optimal cost of \textsc{CkM}. Thus, the obtained solution for \textsc{CkM} has cost at most
\[\frac{\alpha(W \cdot OPT + k \cdot \operatorname{diam}(X))}{W} = \alpha\Bigl(OPT + \frac{k\cdot  \operatorname{diam}(X)}{W}\Bigr) \leq (1+\varepsilon)\alpha \cdot OPT.\]
\end{proof}

Note that since the above theorem holds for $\ell = 2$, and in this setting \pfkmd is a special case of \frc with $k$-median cost, the same hardness result holds for the latter problem too.

\subsection{With Overlapping Groups}
Now, we prove a very strong hardness result for the variant of \pfkmd with overlapping groups, where each point may belong to more than one group. We show that the problem does not admit an $n^{1-\varepsilon}$ approximation for every $\varepsilon > 0$ if $\mathrm{P}\neq \mathrm{NP}$, even when $k=2$. 
\begin{theorem}\label{thm:hardness-overlapping}
It is $\mathrm{NP}$-hard to approximate \pfkmd with overlapping groups within a factor of $n^{1-\varepsilon}$ for every constant $\varepsilon > 0$, even when $k=2$. 
\end{theorem}

\begin{proof}
To prove this hardness result, we will present a reduction from the problem of $2$-coloring a $3$-uniform hypergraph, which is known to be NP-hard~\cite{dinur2005hardness}.
In this problem, given a $3$-uniform hypergraph $H = (V, E)$ with $N = |V|$ vertices, the objective is to decide whether there is a coloring of $V$ in two colors such that no hyperedge in $E$ is monochromatic. 

Consider a $3$-uniform hypergraph $H = (V, E)$. We define an instance of \pfkmd with overlapping groups as follows:
\begin{itemize}
\item The instance has three locations $p_1, q, p_2$, equipped with the line distance: $d(p_1,q) = d(p_2,q) = 1$, and $d(p_1,p_2)=2$.
\item For every vertex $u\in V$, there is a corresponding point $u$ at location $q$, which we will identify with $u$.
\item There are $N^\rho$ points at each of the locations $p_1$ and $p_2$, where $\rho = 2/\varepsilon$.
\item For each hyperedge $e=(v_1,v_2,v_3)\in E$, there is a corresponding group $G_e=\{v_1,v_2,v_3\}$ in the instance. Additionally, there is a group $G_0$ consisting of all the vertices at locations $p_1$ and $p_2$.
\item Let $n = 2N^\rho + N$ be the total number of points. Define the fairness parameter $t$ as $t= n$.
\end{itemize}
Note that the distances between points at the same location are 0. If desired, we can make them strictly positive by making all of them equal to a sufficiently small $0<\delta \ll 1/N^\rho$. This change will not affect the proof below.

For the chosen value of $t$, the fairness constraint simply requires that every non-empty cluster must contain at least one point from each group.
Consider the following two possibilities.

\paragraph{Hypergraph $H$ is $2$-colorable.} Let $V_1$ and $V_2$ be a valid $2$-coloring of $V$. It defines a clustering of the points: $C_1$ consists of all the points at location $p_1$ and all the points (corresponding to the vertices) in $V_1$; similarly, $C_2$ consists of all the points at location $p_2$ and all the points (corresponding to the vertices) in $V_2$.
We open centers for $C_1$ and $C_2$ at locations $p_1$ and $p_2$, respectively. Since each hyperedge $e$ has at least one vertex from $V_1$ and one from $V_2$, both clusters $C_1$ and $C_2$ contain at least one representative from each group $G_e$. Clearly, $C_1$ and $C_2$ also contain representatives from $G_0$. Therefore, clustering $C_1, C_2$ satisfies the fairness constraint. Its cost is $N$, since each vertex at location $q$ contributes 1 and all other points contribute 0 to the cost.

\paragraph{Hypergraph $H$ is not $2$-colorable.} In this case, there is no clustering with two non-empty clusters, satisfying the fairness constraint. If there were such a clustering $C_1, C_2$ then its restriction to points in $V$ would define a valid 2-coloring of $H$.
Therefore, every fair clustering with $k=2$ clusters has only one non-empty cluster; that is, all points belong to the same cluster. It is immediate that the cost of the clustering is at least $2 N^\rho$.

Now, it is NP-hard to distinguish between the possibilities listed above, and accordingly between the cases when the clustering cost is at most $N$ and is at least $2N^\rho$. We conclude that the problem is NP-hard to approximate within $2N^{\rho-1}=\frac{2N^\rho}{N} \geq \frac{n/2}{n^{\varepsilon/2}} \geq n^{1-\varepsilon}$.   
\end{proof}

\section{Experiments}\label{sec:experiments}

Our empirical result builds on the fair representational clustering framework used by Bera et al.~\cite{bera2019fair}. We use their implementation of a \textsc{$k$-Median} approximation to get the initial (not necessarily fair) solution that selects the centers $F$. Then we continue with our new algorithm to find a \pfkmd solution. Of interest is the ``cost of fairness'', the amount paid by our algorithm beyond the original \textsc{$k$-Median} solution's cost.

We also use the same four data sets from the UCI database\footnote{\url{https://archive.ics.uci.edu/datasets}} that were used in \cite{bera2019fair}. However, we do not directly compare our final results with theirs since the notion of fairness we are considering differs from theirs. The main focus of our experiments is evaluating the cost of fairness in practice using our algorithm.

\subsection{Datasets and Setup}
To evaluate the performance of our proposed approach, we performed experiments on four datasets, each comprising both numerical and categorical attributes. The data sets and their relevant attributes are summarized below:

\begin{itemize}
    \item \textbf{bank:} This dataset includes one categorical variable (\textit{marital}, with three unique values) and three numerical variables (\textit{age}, \textit{balance}, and \textit{duration}).
    \item \textbf{adult\_race:} This dataset features one categorical variable (\textit{race}, with five unique values) and three numerical variables (\textit{age}, \textit{final-weight}, and \textit{education-num}).
    \item \textbf{creditcard\_education:} This dataset includes one categorical variable (\textit{education level}, with five unique values) and three numerical variables (\textit{LIMIT\_BAL}, \textit{AGE}, and \textit{BILL\_AMT1}).
    \item \textbf{census1990\_ss:} This dataset has one categorical variable (\textit{dAge}, with eight unique values) and five numerical variables (\textit{dAncstry1}, \textit{dAncstry2}, \textit{iAvail}, \textit{iCitizen}).
\end{itemize}

Each data set is processed to conform to the clustering requirements, ensuring a balanced representation of categorical values while optimizing the numerical attributes. 

\subsection{Experimental Methodology}
The experiments consist of multiple stages, as outlined below:

\begin{enumerate}
    \item \textbf{Balancing Factor $t$:} For each dataset, we run the experiment with fairness parameter $t$ being the minimum possible integer such that the set of points in the input is $t$-balanced. We believe this provides the most interesting information about the cost of fair clustering as it imposes the strictest requirement for the dataset that would still result in a feasible instance.

    \item \textbf{Vanilla \textsc{$k$-Median} Clustering:} Initially, a \textsc{$k$-Median} algorithm is used to identify cluster centers. We employ the single-swap local search based implementation in \cite{bera2019fair}, which provides a $5$-approximation for \textsc{$k$-Median} \cite{AryaGKMMP-SIAMJ04}. This algorithm also serves as the baseline for our evaluations. 
    The results from this step are visualized in the plots labeled as \textit{vanilla}.
    We also note that in all our experiments, the vanilla clustering algorithm had some centers not receiving points from any one of the groups.

    \item \textbf{Fairness Adjustment:} One of the bottlenecks in our algorithm's running time is solving $O(n^2)$ linear programs, one for each distance $D$ between points. To avoid solving so many LPs, we try a smaller number of values $D$. Namely, if $\delta$ denotes the minimum nonzero distance between points in the metric then we try all values of the form $\delta \cdot 1.1^j$ for $j \geq 0$ until this exceeds the maximum distance in the metric.
    This results in only a small number of linear programs being solved in each dataset. One can show that this would only worsen the theoretical worst-case approximation guarantee by a factor $1.1$, but we expect there would be minimal loss to the guarantee in practice. We use \cite{gurobi} in our implementation to solve the LPs.
    We also point out that this step could be further sped up in a parallel setting as different values of $D$ can be handled independently.

    \item \textbf{Finding best assignment using another integer program:}
    After our implementation of our \pfkmd algorithm completes, there is one final thing that could be done to possibly improve the final solution in practice. Namely, our algorithm assigns a particular number of clients from each group to each center $F$. We then use a post-processing step to find the minimum-cost assignment of clients to centers with these given clients-per-center values.
    Although this step does not improve the theoretical approximation guarantee, we observed that it consistently reduces the cost of the final solution in practice.        
\end{enumerate}
\begin{figure}[!t]
    \centering
    \captionsetup[subfigure]{labelformat=empty} 
    \parbox{\textwidth}{\centering (a) $k = 5$.\par}
    \begin{subfigure}{0.24\textwidth}
        \centering
        \includegraphics[width=\linewidth]{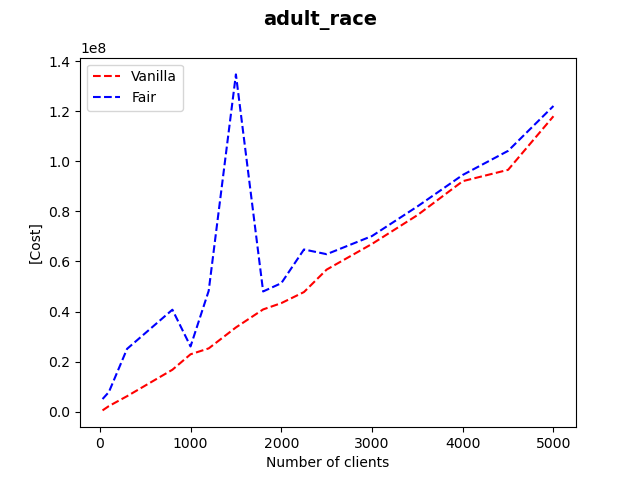}
    \end{subfigure}
    \begin{subfigure}{0.24\textwidth}
        \centering
        \includegraphics[width=\linewidth]{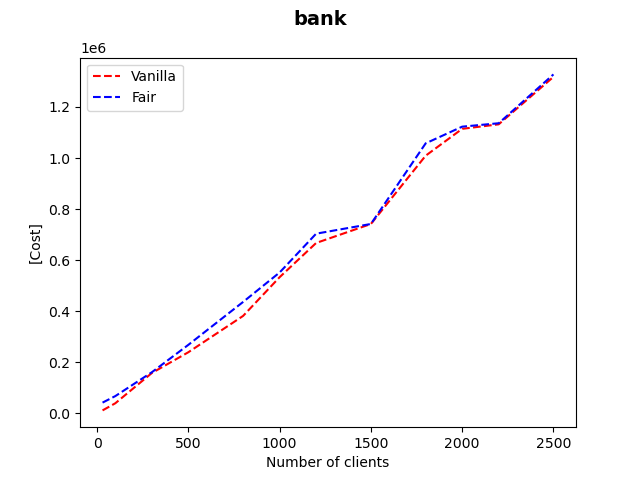}
    \end{subfigure}
    \begin{subfigure}{0.24\textwidth}
        \centering
        \includegraphics[width=\linewidth]{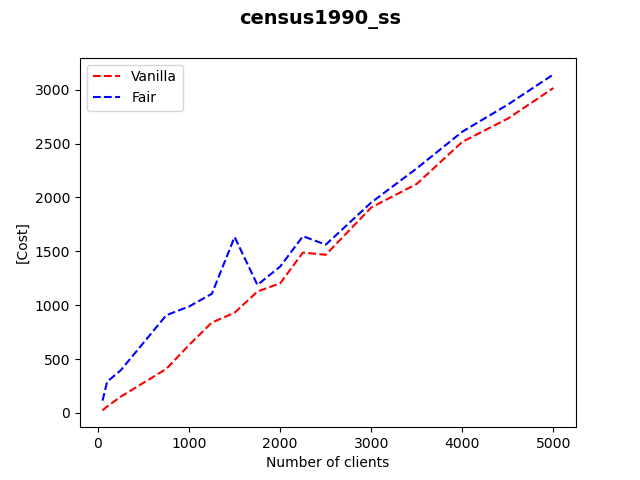}
    \end{subfigure}
    \begin{subfigure}{0.24\textwidth}
        \centering
        \includegraphics[width=\linewidth]{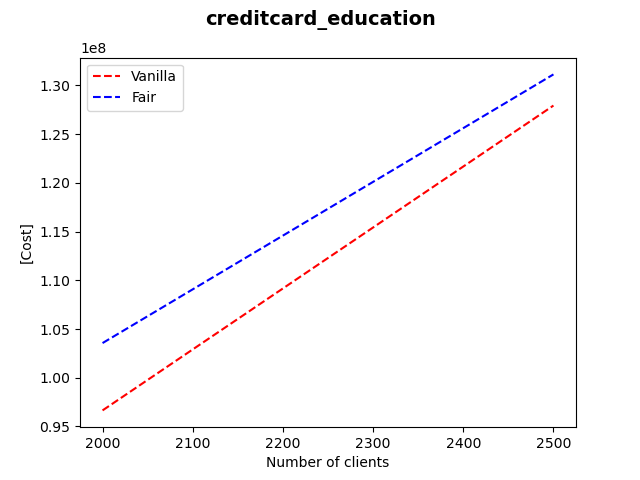}
    \end{subfigure}
    \parbox{\textwidth}{\centering (b) $k = 10$.\par}
    \begin{subfigure}{0.24\textwidth}
        \centering
        \includegraphics[width=\linewidth]{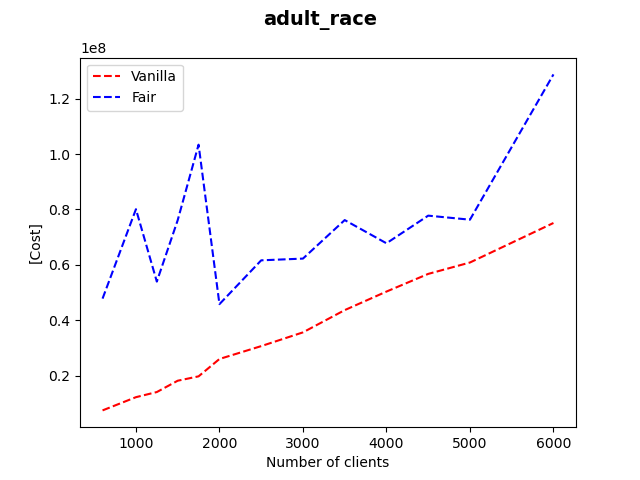}
    \end{subfigure}
    \begin{subfigure}{0.24\textwidth}
        \centering
        \includegraphics[width=\linewidth]{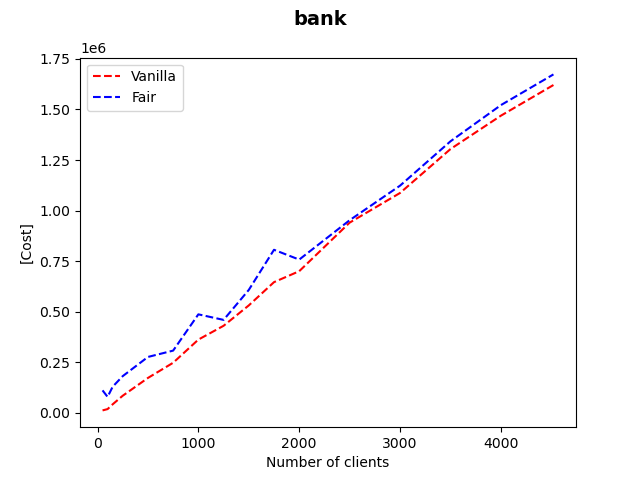}
    \end{subfigure}
    \begin{subfigure}{0.24\textwidth}
        \centering
        \includegraphics[width=\linewidth]{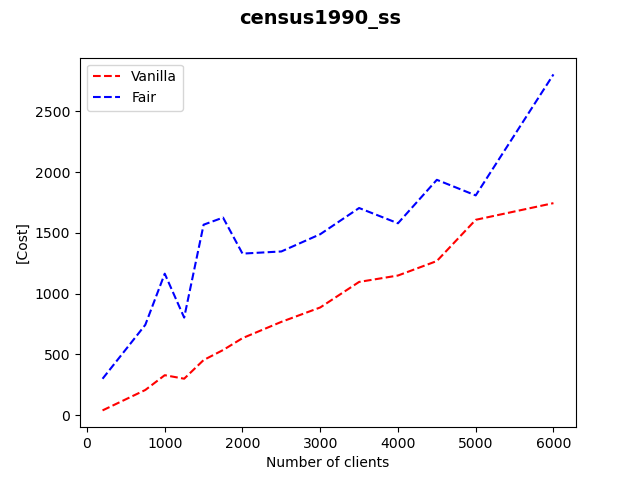}
    \end{subfigure}
    \begin{subfigure}{0.24\textwidth}
        \centering
        \includegraphics[width=\linewidth]{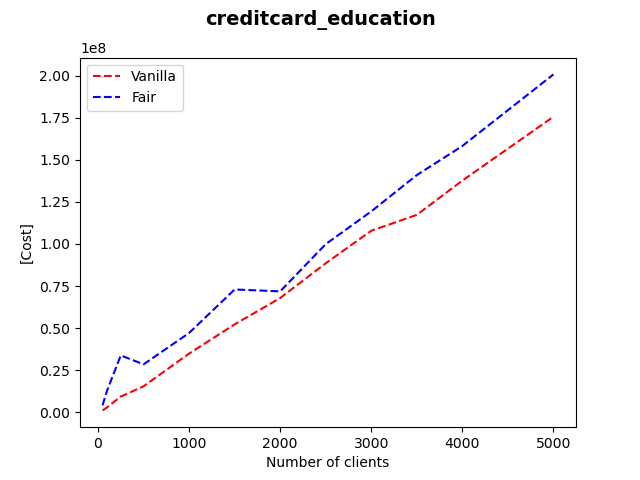}
    \end{subfigure}

    \parbox{\textwidth}{\centering (b) $k = 15$.\par}
    \begin{subfigure}{0.24\textwidth}
        \centering
        \includegraphics[width=\linewidth]{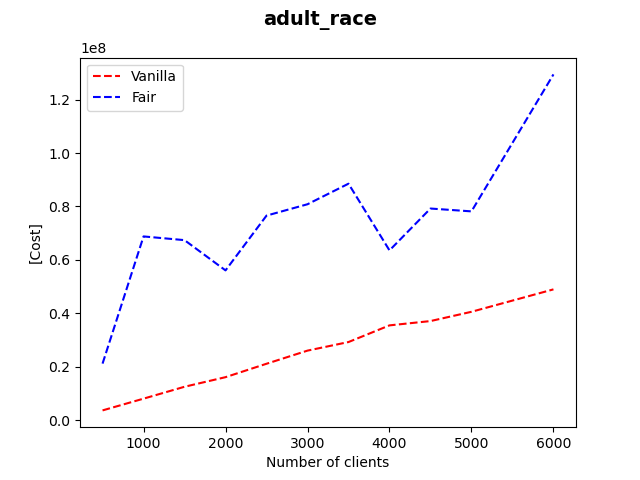}
    \end{subfigure}
    \begin{subfigure}{0.24\textwidth}
        \centering
        \includegraphics[width=\linewidth]{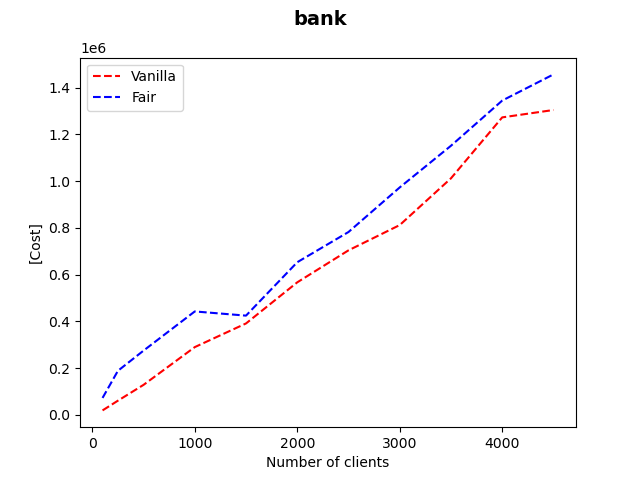}
    \end{subfigure}
    \begin{subfigure}{0.24\textwidth}
        \centering
        \includegraphics[width=\linewidth]{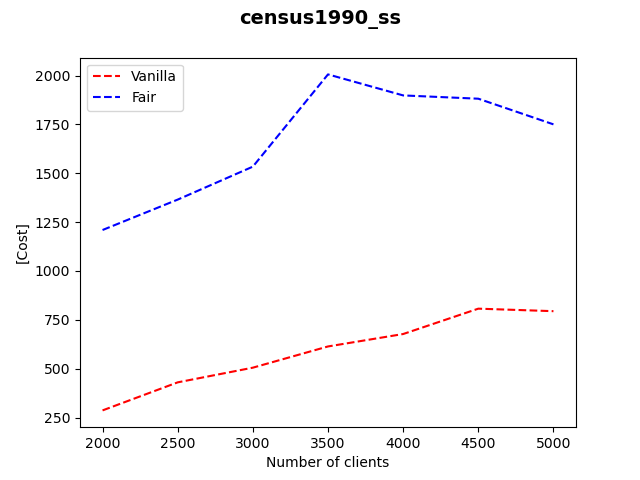}
    \end{subfigure}
    \begin{subfigure}{0.24\textwidth}
        \centering
        \includegraphics[width=\linewidth]{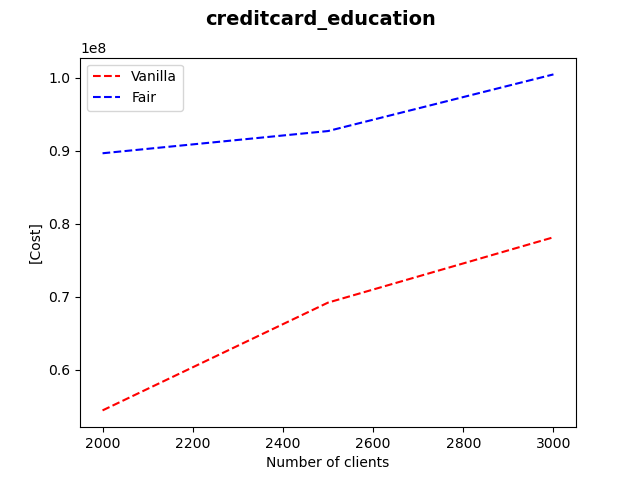}
    \end{subfigure}
    
    \parbox{\textwidth}{\centering (b) $k = 20$.\par}
    \begin{subfigure}{0.24\textwidth}
        \centering
        \includegraphics[width=\linewidth]{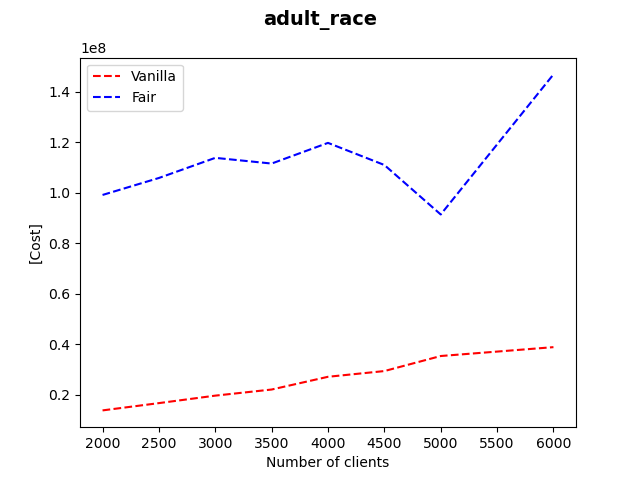}
    \end{subfigure}
    \begin{subfigure}{0.24\textwidth}
        \centering
        \includegraphics[width=\linewidth]{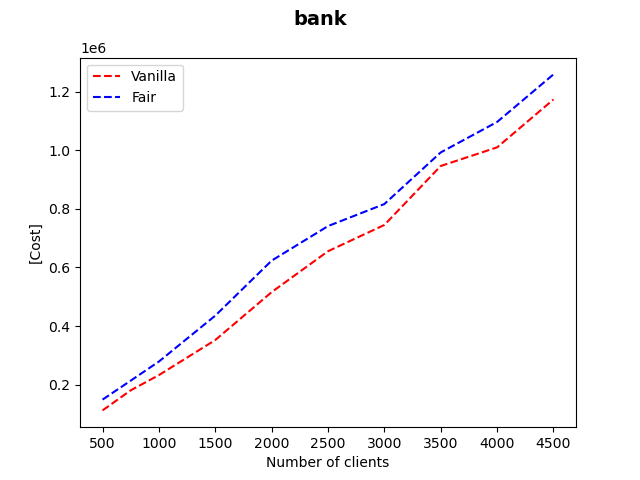}
    \end{subfigure}
    \begin{subfigure}{0.24\textwidth}
        \centering
        \includegraphics[width=\linewidth]{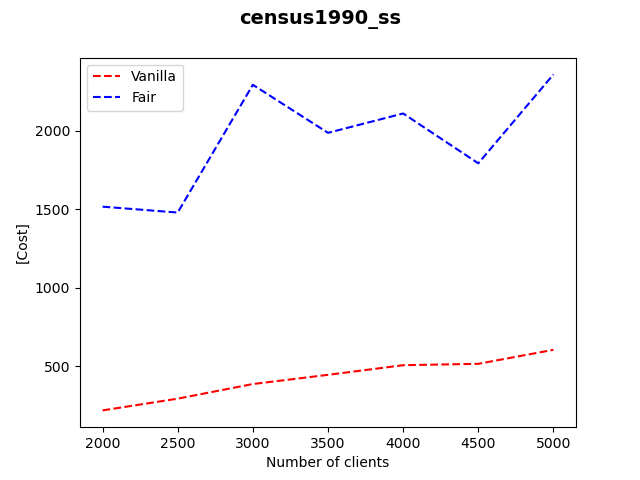}
    \end{subfigure}
    \begin{subfigure}{0.24\textwidth}
        \centering
        \includegraphics[width=\linewidth]{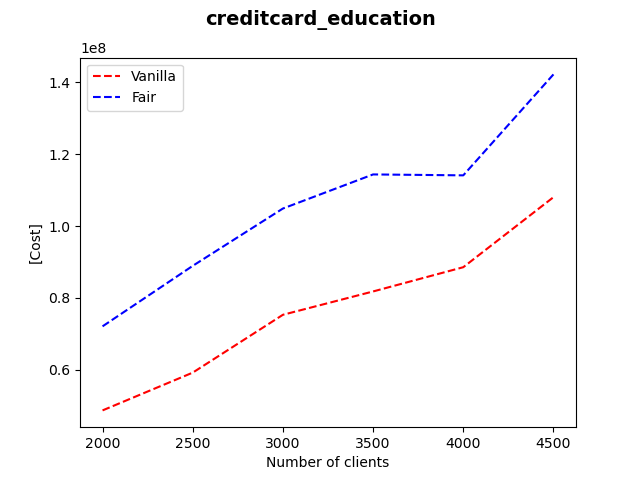}
    \end{subfigure}
    \caption{Cost comparison across multiple datasets with varying target numbers of clusters ($k$).}
    \label{fig:cost}
\end{figure}
\subsection{Results and Analysis}
We conducted experiments on datasets with varying cluster counts ($k$) to examine the impact of fairness adjustments on clustering costs. Figure~\ref{fig:cost} compares the results of vanilla and fair clustering for each dataset. Notably, the difference becomes more pronounced as the number of clusters ($k$) increases, which is expected since balancing cluster centers becomes more challenging for larger $k$. Nevertheless, the cost difference remains significantly better than the worst-case approximation guarantee stated in Theorem~\ref{thm:algopfc}.
\begin{figure}[!t]
    \centering
    \captionsetup[subfigure]{labelformat=empty} 
    \parbox{\textwidth}{\centering (a) $k = 5$.\par}
    \begin{subfigure}{0.24\textwidth}
        \centering
        \includegraphics[width=\linewidth]{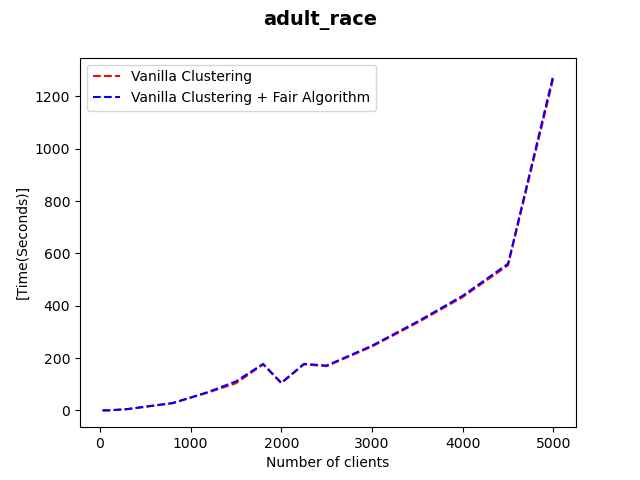}
    \end{subfigure}
    \begin{subfigure}{0.24\textwidth}
        \centering
        \includegraphics[width=\linewidth]{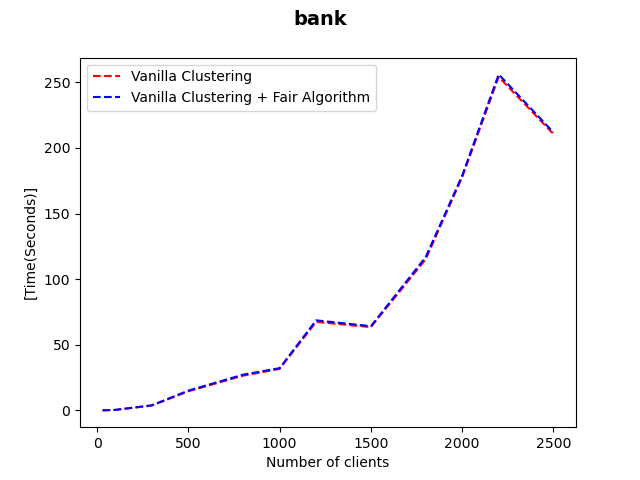}
    \end{subfigure}
    \begin{subfigure}{0.24\textwidth}
        \centering
        \includegraphics[width=\linewidth]{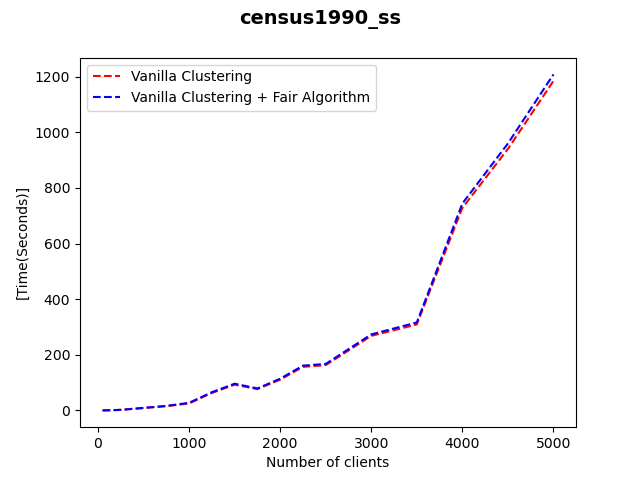}
    \end{subfigure}
    \begin{subfigure}{0.24\textwidth}
        \centering
        \includegraphics[width=\linewidth]{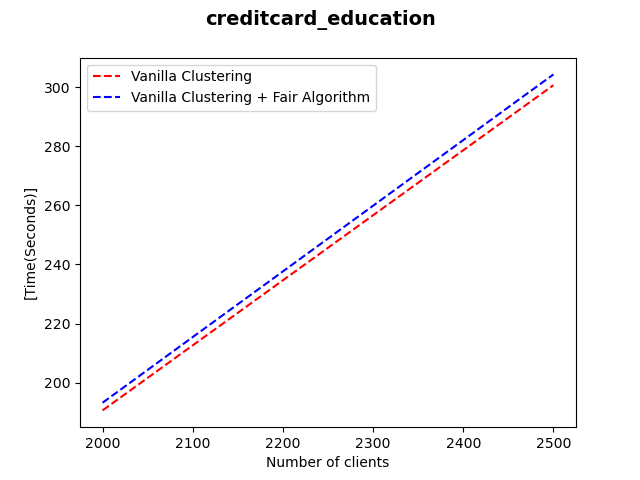}
    \end{subfigure}
    \parbox{\textwidth}{\centering (b) $k = 10$.\par}
    \begin{subfigure}{0.24\textwidth}
        \centering
        \includegraphics[width=\linewidth]{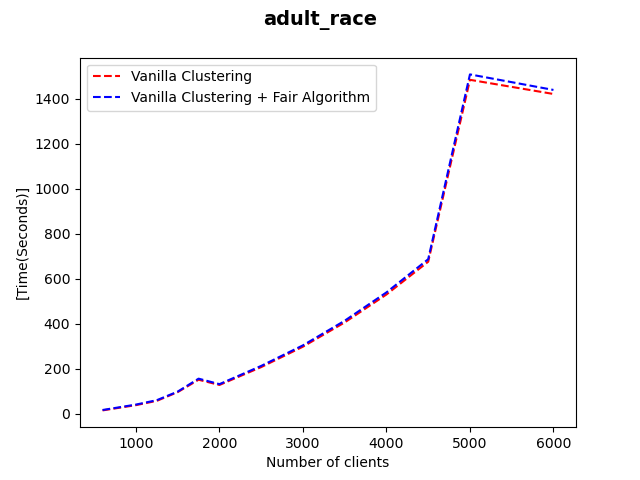}
    \end{subfigure}
    \begin{subfigure}{0.24\textwidth}
        \centering
        \includegraphics[width=\linewidth]{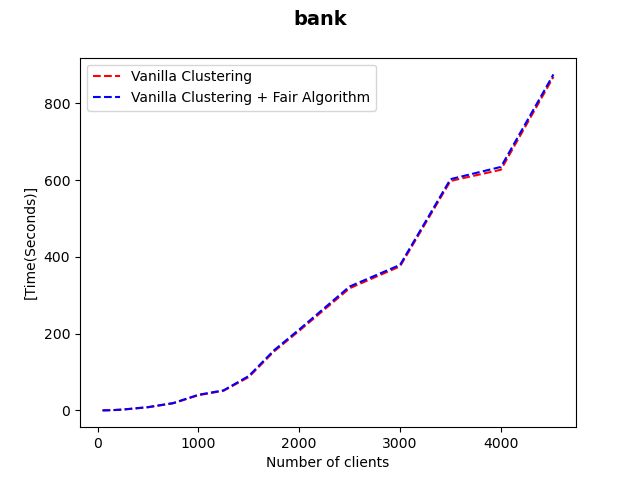}
    \end{subfigure}
    \begin{subfigure}{0.24\textwidth}
        \centering
        \includegraphics[width=\linewidth]{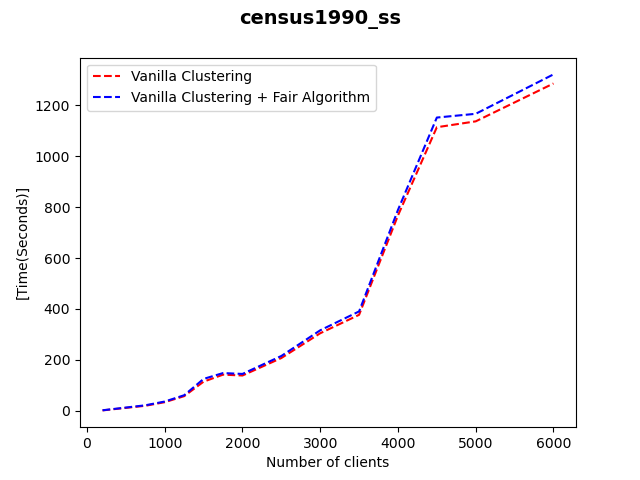}
    \end{subfigure}
    \begin{subfigure}{0.24\textwidth}
        \centering
        \includegraphics[width=\linewidth]{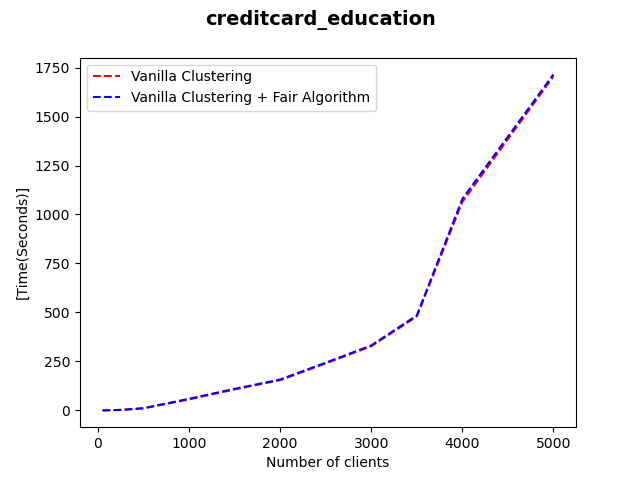}
    \end{subfigure}

    \parbox{\textwidth}{\centering (b) $k = 15$.\par}
    \begin{subfigure}{0.24\textwidth}
        \centering
        \includegraphics[width=\linewidth]{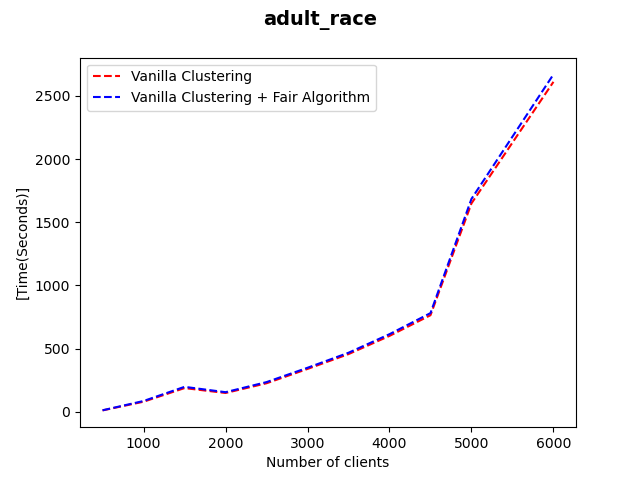}
    \end{subfigure}
    \begin{subfigure}{0.24\textwidth}
        \centering
        \includegraphics[width=\linewidth]{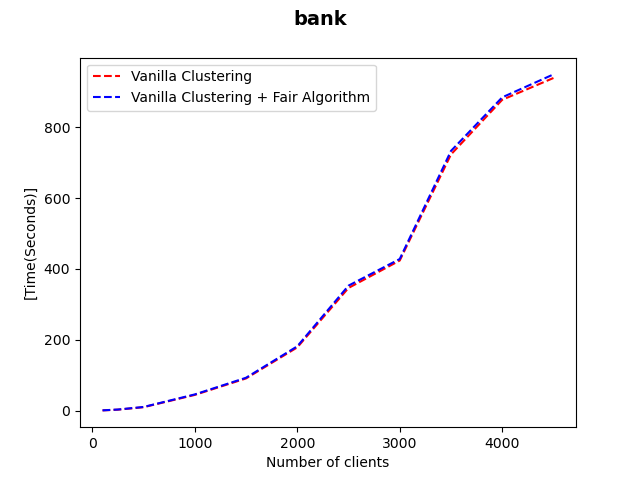}
    \end{subfigure}
    \begin{subfigure}{0.24\textwidth}
        \centering
        \includegraphics[width=\linewidth]{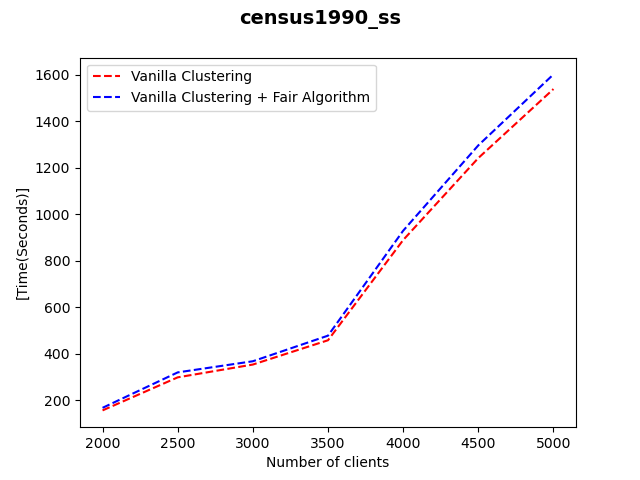}
    \end{subfigure}
    \begin{subfigure}{0.24\textwidth}
        \centering
        \includegraphics[width=\linewidth]{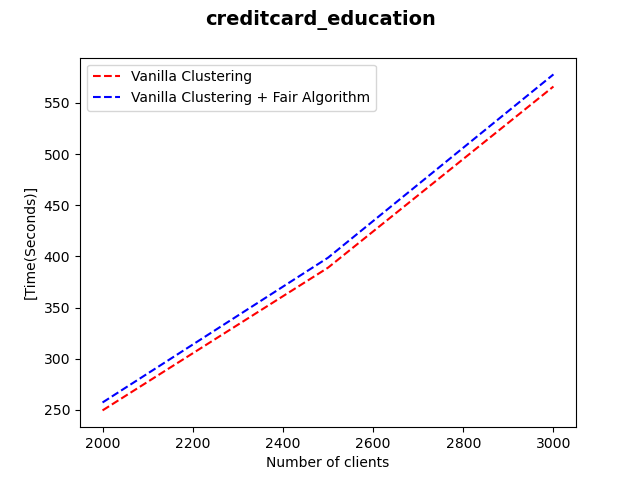}
    \end{subfigure}
    
    \parbox{\textwidth}{\centering (b) $k = 20$.\par}
    \begin{subfigure}{0.24\textwidth}
        \centering
        \includegraphics[width=\linewidth]{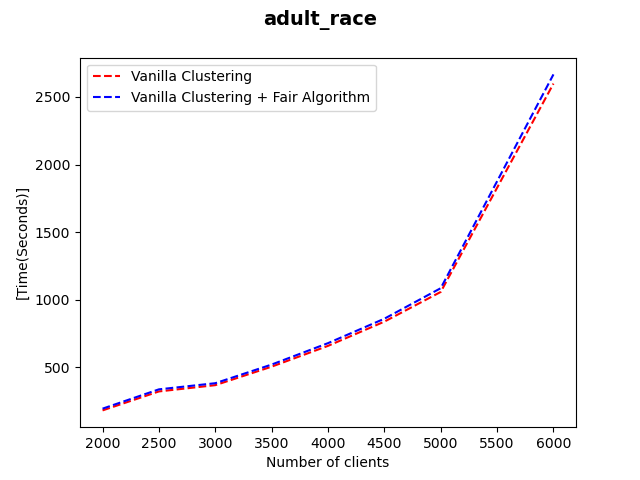}
    \end{subfigure}
    \begin{subfigure}{0.24\textwidth}
        \centering
        \includegraphics[width=\linewidth]{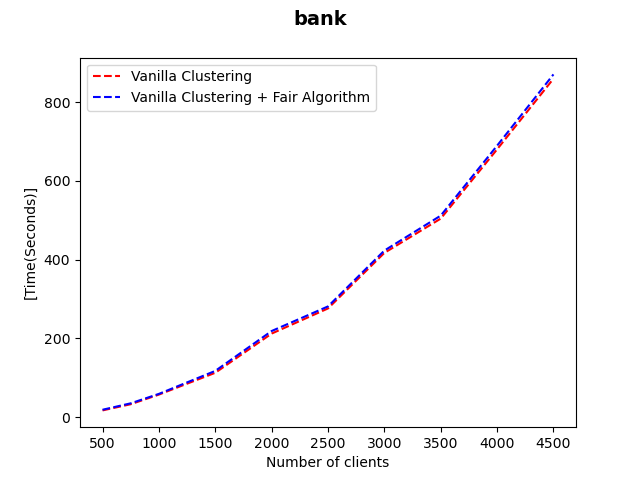}
    \end{subfigure}
    \begin{subfigure}{0.24\textwidth}
        \centering
        \includegraphics[width=\linewidth]{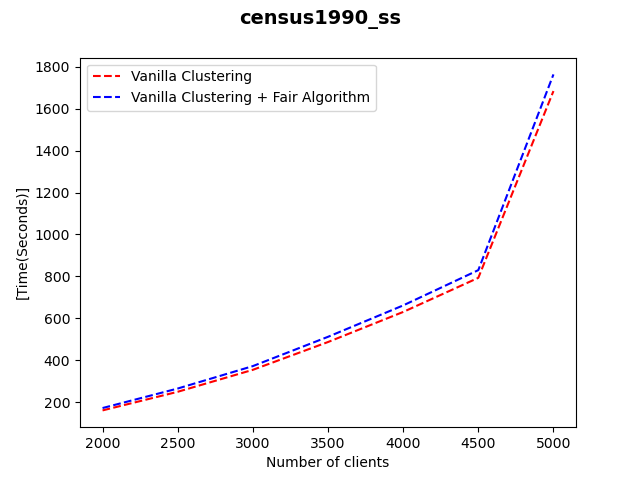}
    \end{subfigure}
    \begin{subfigure}{0.24\textwidth}
        \centering
        \includegraphics[width=\linewidth]{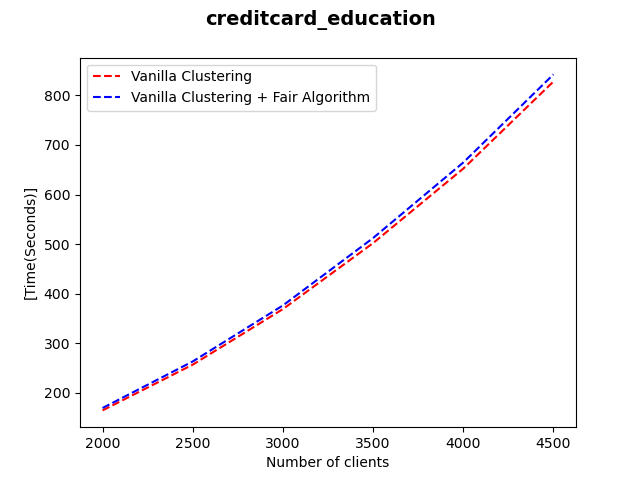}
    \end{subfigure}
    \caption{Runtime comparison across multiple datasets with varying numbers of clusters ($k$).}
    \label{fig:time}
\end{figure}

Plots of the running time of vanilla clustering versus our algorithm are presented in Figure~\ref{fig:time}. In all cases, it is clear that the
bottleneck is the vanilla clustering routine and the overall algorithm could be sped up by using a faster initial \textsc{$k$-Median}
approximation.

\bibliographystyle{plain}
\bibliography{main}
\end{document}